\documentclass[11pt,twoside,a4paper]{article}
\usepackage[affil-it]{authblk}
\usepackage{tikz}
\usetikzlibrary{matrix}
\usepackage[cp1250]{inputenc}
\usepackage[english]{babel}
\usepackage[toc,page]{appendix}
\usepackage[colorlinks,citecolor=green,urlcolor=blue,bookmarks=false,hypertexnames=true]{hyperref}

\usepackage{amssymb,amsmath,amsfonts,amsthm,lscape,xcolor,color,enumerate}
\usepackage[color,all]{xy}

\input xy
\xyoption{all}
\xyoption{frame}

\def\sT{\mathsf T}
\newdir{ (}{{}*!/-5pt/@^{(}}
\newdir{|>}{!/4,5pt/@{|}*:(1,-.2)@^{>}*:(1,+.2)@_{>}}

\def\d{\mathrm d}

\def\J2{\mathsf J^2}

\newtheorem{theorem}{Theorem}

\newtheorem{lemma}[theorem]{Lemma}

\title{The globalization problem of the Hamilton-DeDonder-Weyl equations on a local $k$-symplectic framework}

\begin{document}

\author{O\u{g}ul Esen$^{\dagger}$, Manuel de Le\'on$^{\ddagger}$, Cristina
Sard\'on$^{*}$,  Marcin Zajac$^{**}$}
\date{}
\maketitle

\centerline{Department of Mathematics$^{\dagger}$,}
\centerline{Gebze Technical University, 41400 Gebze, Kocaeli, Turkey.}
\centerline{oesen@gtu.edu.tr}
\vskip 0.5cm

\centerline{Instituto de Ciencias Matem\'aticas, Campus Cantoblanco$^{\ddagger}$,}
\centerline{Consejo Superior de Investigaciones Cient\'ificas}
\centerline{and}
\centerline{
Real Academia Espa{\~n}ola de las Ciencias.}
\centerline{C/ Nicol\'as Cabrera, 13--15, 28049, Madrid, Spain.}
\centerline{mdeleon@icmat.es}
\vskip 0.5cm

\centerline{Instituto de Ciencias Matem\'aticas, Campus Cantoblanco$^{*}$,}
\centerline{Consejo Superior de Investigaciones Cient\'ificas.}
\centerline{C/ Nicol\'as Cabrera, 13--15, 28049, Madrid, Spain.}
\centerline{cristinasardon@icmat.es}
\vskip 0.5cm

\centerline{Department of Mathematical Methods in Physics$^{**}$,}
\centerline{Faculty of Physics. University of Warsaw,}
\centerline{ul. Pasteura 5, 02-093 Warsaw, Poland.}
\centerline{
marcin.zajac@fuw.edu.pl}

\begin{abstract}
In this paper we aim at addressing the globalization problem of Hamilton-DeDonder-Weyl equations on a local $k$-symplectic framework and we introduce the notion of {\it locally conformal $k$-symplectic (l.c.k-s.) manifolds}. This formalism describes the dynamical properties of physical systems that locally behave like multi-Hamiltonian systems. Here, we describe the local Hamiltonian properties of such systems, but we also provide a global outlook by introducing the global Lee one-form approach. In particular, the dynamics will be depicted with the aid of the Hamilton--Jacobi equation, which is specifically proposed in a l.c.k-s manifold.
\end{abstract}

\tableofcontents

\section{Introduction}

In this work, we address the globalization problem of local Hamiltonian-De Donder-Weyl (HDW) dynamics on $k$-symplectic spaces, and we propose the associated Hamilton-Jacobi theory on $k$-locally conformal symplectic manifolds. In the classical symplectic framework, the glueing problem of local Hamiltonian dynamics was investigated by Vaismann \cite{vaisman} in terms of locally conformal symplectic (l.c.s.) manifolds, being the latter introduced by Lee in \cite{Hwa}. 
Our interest in this work is to generalize this approach to the case of $k$-symplectic manifolds. In that direction, 
we shall propose a new geometric structure, which we will call a locally conformal $k$-symplectic (l.c.k-s.) structure, to provide a globalization of the local HDW equations. Further, we shall propose a Hamilton-Jacobi formalism for the dynamics in such manifold. Let us now recall some introductory concepts to follow along the present work and to state our goal more explicitly.

\subsection{Locally conformal symplectic manifolds}
Consider an even dimensional manifold $M$ equipped with a non-degenerate two-from $\omega$. The manifold $M$ is called a locally conformal symplectic (l.c.s.) if $M$ has an open cover $\{U_\alpha\}$ and there exists a family of smooth functions $\sigma_\alpha$ on each chart such that $d(e^{-\sigma_\alpha}\omega\vert_\alpha)$ vanishes identically \cite{Banzoni,vaisman}. Here, $\omega\vert_\alpha$ denotes the restriction of the two-form $\omega$ to the chart $U_\alpha$. In this picture, the local two-form defined by $\omega_\alpha:=e^{-\sigma_\alpha}\omega\vert_\alpha$ turns out to be symplectic and the pair $(U_\alpha,\omega_\alpha)$ becomes a symplectic manifold.

Referring to this symplectic framework, for a local Hamiltonian function $h_\alpha$, the local geometric Hamilton equation is
\begin{equation} \label{locHE}
\iota_{X_\alpha}\omega_\alpha=dh_\alpha.
\end{equation} 

Notice that the Hamiltonian vector fields $X_\alpha$ are still the same even if we multiply the local Hamilton equation \eqref{locHE} by scalars $\lambda_{\beta\alpha}$, in overlapping local charts $U_\alpha$ and $U_\beta$. These scalars can be written as $\lambda_{\beta\alpha}=e^{\sigma_\alpha-\sigma_\beta}$ for local functions $\sigma_\alpha$'s. This leads to the determination of a line bundle $L$ over the manifold $M$. Accordingly, we can glue up the local Hamiltonian functions according to the following transition relation   
\begin{equation} \label{coc-h}
h_\beta=e^{\sigma_\alpha-\sigma_\beta}h_\alpha. 
\end{equation}
This family defines a line bundle valued function $\tilde{h}$ on $M$. Further, we define local functions $h\vert_{\alpha}:=e^{\sigma_\alpha}h_\alpha$ which can be considered as a local picture of a real valued global Hamiltonian function $h$. 
Let us now discuss how we can formulate geometrically the dynamics generated by $h$, not necessarily relying on coordinates. For this, we address an alternative definition of l.c.s. manifolds. 
\bigskip


\noindent \textbf{Lee form.} An equivalent definition of a l.c.s. manifold is possible with the aid of
a compatible one form by glueing up the exterior derivatives $d\sigma_\alpha$ of the local functions $\sigma_\alpha$ to a well defined closed one-form $\theta$, called the Lee one-form \cite{Hwa}.  This observation leads us to an alternative definition of l.c.s. manifolds in terms of global differential forms. An even dimensional manifold $M$ endowed with a non-degenerate two-form $\omega$ is called a l.c.s. manifold if there exists a globally defined one-form $\theta$ such that
\begin{equation}\label{lcs1}
d\omega=\theta\wedge\omega , \quad d\theta=0.
\end{equation}
According to this definition, we denote a l.c.s. manifold by a triple $(M,\omega,\theta)$ where $\omega$ is an almost symplectic two-form and $\theta$ is a Lee form satisfying (\ref{lcs1}). This realization of locally conformal symplectic manifolds implies that a l.c.s. manifold is symplectic if and only if the Lee form $\theta$ vanishes identically. 
\bigskip

\noindent \textbf{Lichnerowicz-de Rham differential.} Let $\theta$ be a closed one-form on a manifold $M$. The Lichnerowicz-de Rham differential $d_\theta$ is defined on the space of differential forms $\Lambda(M)$ 
\begin{equation} \label{LdR}
d_\theta: \Lambda^k(M) \rightarrow \Lambda^{k+1}(M) : \beta \mapsto d\beta-\theta\wedge\beta,
\end{equation}
where $d$ denotes the exterior (de Rham) derivative \cite{GuLi84, HaRy99}. A triple $(M,\omega,\theta)$ involving a non-degenerate two-form $\omega$ and a closed one-form $\theta$ is a l.c.s. manifold if and only if $d_\theta\omega=0$. Recall that if we employ the local equation of dynamics in (\ref{locHE}), and we use the definitions of $\omega_\alpha$ and $h_\alpha$, then, glueing up the local one-forms $d\sigma_\alpha$ to the global Lee one-form $\theta$, we arrive at
\begin{equation} \label{HamEq-lcs}
\iota_{X_h}\omega=d_\theta h.
\end{equation}
Here, the Hamiltonian vector field $X_h$ is the global realization of the local vector fields $X_\alpha$ in (\ref{locHE}). 

\bigskip

\noindent \textbf{L.c.s. structures on cotangent bundles.} Consider the canonical symplectic manifold $T^*Q$ equipped with the canonical symplectic two-form $\Omega_Q=-d\Theta_Q$ where $\Theta_Q$ is the canonical one-form on $T^*Q$ \cite{AbMa78,Arnold,LeRo89}. Let $\vartheta$  be a closed form on $Q$ and pull it back to $T^*Q$ by means of the cotangent bundle projection $\pi_Q$. This gives us a closed semi-basic one-form $\theta=\pi_Q^*(\vartheta)$. By means of the Lichnerowicz-deRham differential, we define a two-form 
\begin{equation} \label{omega_theta}
\Omega_\theta=-d_\theta(\Theta_Q)= -d\Theta_Q+\theta\wedge \Theta_Q=\Omega_Q+\theta\wedge \Theta_Q
\end{equation}
on the cotangent bundle $T^*Q$.  Notice that the triple 
\begin{equation} \label{T*_Q}
T^*_\theta Q:=(T^*Q,\Omega_\theta,\theta), \qquad \Omega_\theta=-\d_\theta\Theta_Q
\end{equation} 
determines a l.c.s. manifold with the Lee-form $\theta$, see \cite{ChMu17,OtSt15}. In short, we denote this l.c.s. manifold by $T^*_\theta Q$. Note that all l.c.s. manifolds are locally $T^*_\theta Q$ for some $Q$ and for a closed one-form $\vartheta$.

\subsection{Geometric Hamilton-Jacobi theory}

Consider a Hamiltonian system $(T^*Q,\Omega_Q,h)$ generated by a Hamiltonian function $h$. In the Hamilton-Jacobi theory one aims at finding a function $W$, the so-called generating function, satisfying the Hamilton-Jacobi equation (HJE)
\begin{equation}\label{HJeq1}
h\left(q^i,\frac{\partial W}{\partial q^i}\right)=\epsilon,
\end{equation}
where $\epsilon$ is a constant \cite{Arnold}. The geometrical picture is the following \cite{CaGrMaMaMuRo06}. Let $X_{h}$ be the Hamiltonian vector field for the Hamiltonian function $h$ then, consider the following projected vector field 
\begin{equation}
 X_h^{dW}=T\pi_Q\circ X_h\circ dW,
\end{equation} 
on $Q$. Here, $T{\pi}_Q$ is the tangent mapping of the cotangent bundle projection $\pi_Q$. The following theorem asserts that an integral curve of $X_{h}^{dW}$ can be lifted to an integral curve of $X_{h}$ if and only if $W$ is a solution of \eqref{HJeq1}:
\begin{theorem}
\label{HJT-Thm} For a closed one-form $dW$ on $Q$ the following conditions are
equivalent:
\begin{enumerate}
\item The vector fields $X_{h}$ and $X_{h}^{dW}$ are $dW$-related.
\item $d\left( h\circ dW \right)=0.$
\end{enumerate}
\end{theorem}
\bigskip

\noindent \textbf{A Hamilton-Jacobi theory in l.c.s. manifolds.} Recently, the authors of this same paper have proposed a Hamilton--Jacobi theory on a l.c.s. manifold \cite{EsLeSaZa19}. Let us summarize here the results we obtained. Consider a Hamiltonian vector field $X_h$ defined through the equation (\ref{HamEq-lcs}). Consider now a section $\gamma$ of the cotangent bundle and define a vector field on $Q$ as
\begin{equation} \label{X_H-gamma}
 X_{h}^{\gamma}=T\pi\circ X_{h}\circ\gamma.
\end{equation}
One can define the vector field $X_{h}^{\gamma}$ through the following commutative diagram
\[
\xymatrix{ T_{\theta}^*Q
\ar[dd]^{\pi_Q} \ar[rrr]^{X_{h}}&   & &TT_{\theta}^*Q  \ar[dd]^{T\pi_Q}\\
  &  & &\\
Q \ar@/^2pc/[uu]^{\gamma}\ar[rrr]^{X_{h}^{\gamma}}&  & & TQ}
\]
Now, the Hamilton-Jacobi theorem for l.c.s. cotangent bundles reads as follows. 
\begin{theorem}
\label{HJT-Thm-lcs} Consider a one-form $\gamma$ whose image is a Lagrangian submanifold of the locally conformal symplectic manifold $T^*_\theta Q$ with respect to the almost symplectic two-form $\Omega_\theta$, that is $d_\vartheta \gamma=0$. Then, the following  conditions are equivalent:
\begin{enumerate}
\item The vector fields $X_{h}$ and $X_{h}^{\gamma}$ are $\gamma$-related,
\item $ d_\vartheta (h\circ\gamma)=0.$
\end{enumerate}
\end{theorem}

\subsection{The goal and the contents}

The aim of this work is two-fold. The first aim, rooting in Vaisman's discussion in \cite{vaisman}, is to address the globalization problem of local HDW equations on local $k$-symplectic spaces. Such discussion will lead us to introduce the notion of what we call henceforth locally conformal $k$-symplectic (l.c.k-s.) manifolds. Our second goal is to present a geometric Hamilton-Jacobi formalism on l.c.k-s. manifolds. This work can be considered as a continuation or complementary to \cite{EsLeSaZa19}. 
\bigskip

\noindent This paper contains three main sections. In the first section, we recall the notion of $k$-symplectic manifolds, HDW dynamics, and the Hamilton-Jacobi theorem on $k$-symplectic manifolds presented in \cite{LeDiMaSaVi10}. We present a coordinate free proof of this Hamilton-Jacobi theorem. The next section \ref{Sec-lcks} concerns the introduction of our novel concept, here presented as locally conformal $k$-symplectic manifolds. We formulate Hamiltonian dynamics on l.c.k-s. manifolds along with a Hamilton-Jacobi theorem, which constitutes our second main result in this paper. In the last section \ref{Exp}, we present a non-trivial example to illustrate the previous theoretical discussion.  All manifolds are considered to be real, paracompact, connected and $C^\infty$. All
 maps are $C^\infty$. Sum over crossed repeated indices is understood.

\section{k-symplectic framework} \label{Sec-k-symp}

\subsection{k-tangent bundle}

Let $M$ be an $n$-dimensional manifold. The tangent bundle $(TM,\tau_M,M)$ is a $2n$-dimensional manifold. Here, $\tau_M$ is the tangent bundle projection. A section of the tangent bundle is a vector field $X$ on $M$. The space of smooth vector fields $\mathfrak{X}(M)$ on $M$ has a module structure on the space of smooth functions on $M$. 
\bigskip

\noindent The Whitney sum of $k$ copies of tangent bundles
\begin{equation}\label{TkM}
T^k M:=TM\oplus_M \dots \oplus_M TM,
\end{equation}
is a $nk+n$ dimensional manifold, and the triple $(T^k M,\boldsymbol{\tau}_M,M)$ is a fiber bundle \cite{LSV}. Here, $\boldsymbol{\tau}_M$ is the projection from $T^k M$ to its base manifold $M$. Note that we reserve the bold symbol $\boldsymbol{\tau}_M$ in order to differentiate this projection from the tangent bundle projection $\tau_M$. We refer to  $T^k M$ as the $k$-tangent bundle of $M$. There exists a family of projections $\tau_\kappa$ (here, $\kappa$ runs from $1$ to $k$) from the $k$-tangent bundle $T^k M$ to the tangent bundle $TM$ which is assumed to be the $\kappa^{th}$-copy $TM$ in \eqref{TkM}. A section of $k$-tangent bundle is a $k$-vector field ${\bf X}$ on $M$. We denote the space of $k$-vector fields on $M$ by $\mathfrak{X}^k(M)$. A $k$-vector field ${\bf X}$ is equivalent to
a family of $k$ vector fields on $M$, accordingly, we denote a $k$-vector field by 
\begin{equation} \label{X-decomp}
{\bf X}:=[X_\kappa]=(X_1, \ldots, X_k).
\end{equation}
Referring to the projections $\tau_\kappa$ we can determine the relation
\begin{equation}
\tau_\kappa \circ {\bf X}=X_\kappa, \qquad \forall \kappa=1,\dots,k.
\end{equation}
According to this, we present the following commutative diagram.
\begin{equation}
\xymatrix{ T^k M
\ar[ddr]^{\boldsymbol{\tau}_M} \ar[rr]^{\tau_\kappa}&& TM  \ar[ddl]_{\tau_M}\\
  \\
&M \ar@/^1pc/[uul]^{\mathbf{X}} \ar@/_1pc/[uur]_{X_\kappa}}
\end{equation}

\noindent\textbf{Integrability of $k$-vector fields.} We denote the space of $k$-vector fields on $M$ by $\mathfrak{X}^k(M)$. As in the classical picture, the integrability of $k$-vector fields can also be discussed equivalently. For this, let $U$ be an open neighborhood containing the origin $\mathbf{0}$ in a $k$-dimensional Euclidean space $\mathbb{R}^k$ endowed with coordinates $\mathbf{t}=(t^\kappa)$. An integral section of a $k$-vector field ${\bf X}$
passing through a point $x$ is a differential map
  $\boldsymbol{\phi}$ from $U$ to the manifold $M$ satisfying that 
\begin{equation}
\boldsymbol{\phi}(\mathbf{0})=x, \qquad 
\boldsymbol{\phi}_{*} \left(\frac{\displaystyle\partial}{\partial
t^\kappa}\right) = X_{\kappa}
 \end{equation}
for all $\kappa$ from $1$ to $k$. Here, the notation $\boldsymbol{\phi}_{*}$ refers to the push-forward operation. 
A $k$-vector field is integrable if there is an integral
section passing through every point of $M$.
\bigskip

\noindent\textbf{$k$-tangent lift.} Consider a differentiable mapping $\varphi$ from a manifold $M$ to another manifold $N$. Then $k$-tangent lift of $\varphi$ is a mapping from the $k$-tangent bundle $T^kM$ to the $k$-tangent bundle $T^kN$, and it is defined to be
\begin{equation} \label{k-lift}
T^k\varphi[X_\kappa(z)]=T^k\varphi (X_1(z),\dots,X_k(z))=(T\varphi(X_1(z)),\dots,T\varphi(X_k(z)) )   
\end{equation}
where $T\varphi$ is the tangent lift of $\varphi$. Here, $X_\kappa(z)$ is an element of $T_zM$, and $[X_\kappa(z)]$ is an element of $T_z^kM$. Notice that the $k$-tangent lift of $\varphi$ satisfies the following commutation rule 
\begin{equation}
   \boldsymbol{\tau}_N \circ T^k\varphi = \varphi \circ  \boldsymbol{\tau}_M
\end{equation}
where $\boldsymbol{\tau}$ is the $k$-tangent bundle projection. Let us provide another  commutative diagram summarizing this:
\begin{equation} \label{diagramm}
\xymatrix{ T^k M \ar[rrrr]^{T^k\varphi}
\ar[dddr]^{\boldsymbol{\tau}_M} \ar[rrd]^{\tau_\kappa}&& && T^kN \ar[rrd]^{\tau_\kappa} \ar[dddr]^{\boldsymbol{\tau}_N} \\&& TM\ar[ddl]_{\tau_M} \ar@{.>}[rrrr]^{T\varphi \qquad  \qquad}  && && TN\ar[ddl]_{\tau_N}
  \\ 
  \\
&M \ar[rrrr]^{\varphi} \ar@/^1pc/[uuul]^{\mathbf{X}} \ar@/_1pc/[uur]_{X^\kappa}&& && N \ar@/_1pc/[uur]_{T\varphi \circ X^\kappa} \ar@/^1pc/[uuul]^{(\varphi^k)_* \mathbf{X}}}
\end{equation}
Note that, the commutation of this diagram additionally enables us to define the $k$-th push forward, denoted by $(\varphi^k)_* \mathbf{X}$, of a $k$-vector field $\mathbf{X}$ by means of the differentiable mapping $\varphi$. Accordingly, the relationship between two $k$-vector fields is defined to be
\begin{equation}
(\varphi^k)_* \mathbf{X} \circ \varphi =T^k \varphi \circ \mathbf{X}.
\end{equation}

\subsection{k-symplectic manifolds}

Consider a $n$ dimensional manifold $M$ with $k$ number of closed two-forms $[\Omega^{\kappa}]$, and $nk$-dimensional
integrable distribution $V$. A triple $(M,[\Omega^{\kappa}],V)$ is called  a $k$-symplectic manifold if the following two conditions are satisfied 
\begin{equation}\label{condksymp}
(i) ~ \Omega^{\kappa}\vert_{ V\times V}=0,~ \forall \kappa=1,\dots k, \qquad
 (ii) ~  \bigcap_{\kappa=1}^{k} \ker\Omega^{\kappa}=\{0\}.
\end{equation}
See \cite{aw,aw2,LeonMeSa88a,LeonMeSa88b,LeonMeSa91,LeonMeSa93} for fundamentals on $k$-symplectic manifolds. The canonical model for a $k$-symplectic manifold is the $k$-cotangent bundle of a manifold. Let us depict this in detail.
\bigskip

\noindent \textbf{Lagrangian submanifolds.} Consider a submanifold $N$ of a $k$-symplectic manifold $(M,[\Omega^{\kappa}],V)$. Then $k$-symplectic dual of the tangent space $T_zN$ is defined as follows 
\begin{equation} \label{sdual}
(T_zN)^\perp=\left \{v\in T_zM:\Omega^1(v,w)=\dots=\Omega^k(v,w)=0, \quad \forall w \in T_zN  \right\}.
\end{equation}
A submanifold $N$ is called isotropic if $T_zN\subset(T_zN)^\perp$ for all $z$ in $N$ \cite{LSV,LV}. $N$ is called coisotropic if $(T_zN)^\perp\subset T_zN$ whereas it is called as Lagrangian if $(T_zN)^\perp= T_zN$, for all $z$. It is important to observe that the submanifold of $M$ integrating the distribution $V$ is a Lagrangian submanifold of $M$. We shall characterize the Lagrangian submanifolds more concretely in the generic example of $k$-symplectic manifolds defined in the upcoming paragraph.
 \bigskip

\noindent \textbf{$k$-cotangent bundle.}  
The $k$-cotangent bundle of a manifold $Q$ is defined to be the Whitney sum of $k$ copies of the cotangent bundle
$T^*Q$ as
\begin{equation} \label{Tk*Q}
T_k^*Q:=T^*Q\oplus_Q\dots\oplus_QT^*Q
\end{equation}
equipped with the canonical projection $\boldsymbol{\pi}_Q $ from $T_k^*Q$ to the base $Q$. Additionally, there exists a family of projections $\pi^\kappa$ from $T_k^*Q$ to the $\kappa^{th}$-copy $T^*Q$ in the decomposition \eqref{Tk*Q}. We define a family of closed two-forms on $T_k^*Q$ by pulling the canonical symplectic two-forms on each $T^*Q$ back to $T_k^*Q$ by means of $\pi^{\kappa}$. We denote the set of such two-forms by $[\Omega^{\kappa}]$. Consider the kernel $V$ of the tangent lift $T\boldsymbol{\pi}_Q$ of the projection. It is an integrable distribution.
Notice that the set 
\begin{equation}
(T_k^*Q,[\Omega^{\kappa}],V)
\end{equation}
determines a $k$-symplectic manifold by satisfying both conditions presented in (\ref{condksymp}). 
\bigskip

\noindent
In Darboux' coordinates $(q^i,p^\kappa_i)$ on $T_k^*Q$, the closed two-forms and the  distribution $V$ are computed to be
\begin{equation} \label{k-symp-s-loc}
\Omega^{\kappa}=  dq^i\wedge dp^\kappa_{i}, \qquad
V=\left\langle\frac{\partial}{\partial p^1_i}, \dots,
\frac{\partial}{\partial p^k_i}\right\rangle.
\end{equation}

It is important to note that this local realization is generic for all $k$-symplectic manifolds. That is, any k-symplectic manifold, at every
point of the manifold, admits a local coordinate system $(q^i,p^\kappa_i)$  such that the closed two-forms and the distribution have the local expressions in \eqref{k-symp-s-loc}, see \cite{aw,LeonMeSa93}.
\bigskip

\noindent \textbf{Sections of $k$-cotangent bundle.} 
A section of a $k$-cotangent bundle is a differentiable mapping $\boldsymbol{\gamma}$ from the base manifold $Q$ to the total space $T_k^*Q$ satisfying that $\boldsymbol{\pi}_Q\circ \boldsymbol{\gamma}=\text{id}$. A decomposition of $\boldsymbol{\gamma}$ is 
\begin{equation}
\boldsymbol{\gamma}=[\gamma^{\kappa}]=(\gamma^1,\dots,\gamma^k)
\end{equation}
where each of $\gamma^{\kappa}$ is a differential one-form on $Q$. Notice that, the relationships between $\boldsymbol{\gamma}$ and the one-forms $[\gamma^{\kappa}]$ are given by
\begin{equation} \label{Gamma-gamma-a}
\pi^\kappa\circ \boldsymbol{\gamma}=\gamma^\kappa,  \qquad \forall \kappa=1,\dots,k.
\end{equation}
The following commutative diagram summarizes the discussions
\begin{equation} \label{pi-a}
\xymatrix{ T_k^*Q
\ar[ddr]^{\boldsymbol{\pi}_Q} \ar[rr]^{\pi^{\kappa}}&& T^*Q \ar[ddl]_{\pi_Q}\\
  \\
&Q \ar@/^1pc/[uul]^{\boldsymbol{\gamma}} \ar@/_1pc/[uur]_{\gamma^{\kappa}}}
\end{equation}
A section $\boldsymbol{\gamma}$ of $\boldsymbol{\pi}_M$ is called a closed section if all of the constitutive differential one-forms are closed. 
We perform the following direct calculation
\begin{equation}\label{zero}
\boldsymbol{\gamma}^*\Omega^{\kappa}=\boldsymbol{\gamma}^*(\pi^{\kappa})^*\Omega_Q=(\pi^{\kappa}\circ \boldsymbol{\gamma})^*\Omega_Q=(\gamma^{\kappa})^*\Omega_Q=-\d\gamma^{\kappa}  
\end{equation}
where we have employed  (\ref{Gamma-gamma-a}) in the third identity. From (\ref{zero}) we can see that $\boldsymbol{\gamma}^*\Omega^{\kappa}$ vanishes identically if and only if each $\gamma^{\kappa}$ is closed. This implies that the image space of a section $\boldsymbol{\gamma}$ of the $k$-cotangent bundle is a Lagrangian submanifold if and only if it is closed.

\subsection{Dynamics on k-symplectic manifolds}

Consider the canonical $k$-symplectic manifold
$(T_k^*Q,[\Omega^{\kappa}],V)$ given in \eqref{Tk*Q}. Let us first introduce some notations. Since $T_k^*Q$ is a manifold, one can define its tangent bundle and cotangent bundle as usual $TT_k^*Q$ and $T^*T_k^*Q$. A vector field on $T_k^*Q$ takes values in the tangent bundle $TT_k^*Q$ whereas a one-form on $T_k^*Q$ takes values in the cotangent bundle $T^*T_k^*Q$. Accordingly, we denote the space of vector fields by $\mathfrak{X}(T_k^*Q)$, and the space of one-forms by $\Lambda^1(T_{k}^{*} Q)$. Further, we can define the $k$-tangent bundle $T^kT_k^*Q$ of the $k$-cotangent bundle $T_k^*Q$. Sections of this fibration are $k$-vector fields on $T_k^*Q$. Recalling (\ref{X-decomp}), we define the following musical mapping 
\begin{equation} \label{flat}
\flat : \mathfrak{X}^k(T_k^*Q)  \longrightarrow \Lambda^1(T_{k}^{*} Q)  : \mathbf{X}=\left(X_{1}, \ldots, X_{k}\right) \mapsto \sum_{\kappa=1}^k \iota_{X_{\kappa}} \Omega^{\kappa}.
\end{equation}
If $k=1$, one arrives at the classical symplectic framework where the musical mapping becomes a symplectic isomorphism. But, for $k>1$, the kernel of the musical mapping is far from being trivial, so it fails to be an isomorphism. 

\bigskip 

\noindent \textbf{Dynamics on the $k$-cotangent bundle.} Start with a $k$-cotangent bundle $(T_k^*Q,[\Omega^{\kappa}],V)$. Consider a smooth Hamiltonian function $H$ defined on $T_k^*Q$. We define the Hamilton-DeDonder-Weyl (HDW) equation  as 
\begin{equation} \label{HamEq-k}
\flat(\mathbf{X})=dH,
\end{equation}
where $\flat$ is the operator defined in (\ref{flat}). Notice that, in this framework, HDW equation determines a $k$-vector field $\mathbf{X}$ on $T_k^*Q$. We call $\mathbf{X}$ a  HDW $k$-vector field. Let us write the HDW equation more explicitly by endowing the decomposition in (\ref{X-decomp}), i.e.,
\begin{equation}
\sum_{\kappa=1}^k\iota_{X_\kappa}\Omega^{\kappa}=d H.
 \label{generic0}
\end{equation}
It is important to note that it is not possible to determine a unique HDW $k$-vector field satisfying the HDW equation (\ref{HamEq-k}). Let us clarify this in local coordinates.
\bigskip

\noindent \textbf{Dynamics in local coordinates.} Recall the existence of Darboux' coordinates $(q^i,p_i^\kappa)$ on $T_k^*Q$, and the local realizations of the closed two-forms $\Omega^{\kappa}$ in (\ref{k-symp-s-loc}). In this chart, each vector field $X_\kappa$ can be written as 
 \begin{equation}
  \label{coorxa}
   X_\kappa =(X_\kappa)^i\frac{\partial}{\partial q^i}+(X_\kappa)_i^\lambda\frac{\partial}{\partial p_i^\lambda},
\end{equation}
where $(X_\kappa)^i$ and $(X_\kappa)_i^\lambda$ are real valued function on $T^*_k Q$. 
Then, we obtain that the HDW equation
(\ref{generic0}) is equivalent to the system of equations
 \begin{equation}
  \label{11}
 \frac{\partial H}{\partial q^i}=- \sum_{\kappa=1}^k\,(X_\kappa)^\kappa_i
, \qquad \frac{\partial H}{\partial p^\kappa_i}=(X_\kappa)^i.
\end{equation}
A $k$-vector field $\mathbf{X}$ on $T^*_k Q$ lies in the kernel of the musical mapping $\flat$ in \eqref{flat} if the components of the constitutive vector fields satisfy 
\begin{equation}
(X_{\kappa})^{i}=0, \qquad \sum_{\kappa=1}^k\left(X_{\kappa}\right)_{i}^{\kappa}=0.
\end{equation}
The existence of solutions of (\ref{11})
 is guaranteed. These depend on
$n(k^2-1)$ arbitrary functions although they are not necessarily integrable,
so the number of integrability conditions imply that
the number of arbitrary functions will be less than $n(k^2-1)$. 
\bigskip

\noindent \textbf{Solutions of HDW equations.} A solution to the HDW equations (\ref{11}) is a mapping $\boldsymbol{\phi}$ from an open neighbourhood $U$ of the Euclidean space $\mathbb{R}^k$ to the base manifold $T_k^*Q$. In Darboux' coordinates $(q^i,p_i^\kappa)$, we write such a solution as 
$\boldsymbol{\phi}(\mathbf{t})=(\phi^i(\mathbf{t}),\phi^\kappa_i(\mathbf{t}))$. In accordance with this, a solution must satisfy 
\begin{equation}
\label{he20}
 \frac{\partial H}{\partial q^i}=- \frac{\partial\phi^\kappa_i} {\partial t^\kappa}
\quad , \quad \frac{\partial H} {\partial
p^\kappa_i}=\frac{\partial\phi^i}{\partial t^\kappa},
\end{equation}
where we have assumed that $\mathbf{t}=(t^\kappa)$ is in $\mathbb{R}^k$. A solution $\boldsymbol{\phi}$ of the HDW-equations
(\ref{he20}) is said to be an {admissible solution} if ${\rm
Im}\,\boldsymbol{\phi}$ is a closed embedded submanifold of $T_k^*Q$. An admissible solution of the HDW-equations is an
integral section of an integrable $k$-vector field ${\bf
X}$ \cite{RRRSV}. We say that the set $(T_k^*Q,[\Omega^{\kappa}],V,H)$  is an admissible $k$-symplectic
Hamiltonian system if all solutions of the HDW-equation are
admissible.

\subsection{Hamilton--Jacobi theory} \label{HJ-k}

Let us briefly recall the Hamilton-Jacobi theory in the $k$-symplectic setting. One can find more detailed discussions on the issue e.g. in \cite{LeDiMaSaVi10,LSV}. Consider an admissible $k$-symplectic
Hamiltonian system $(T_k^*Q,[\Omega^{\kappa}],V,H)$.  In this framework, the Hamilton-Jacobi problem is the following partial differential equation 
\begin{equation}
H\left(q^{i}, \frac{\partial W^{1}}{\partial q^{i}}, \ldots, \frac{\partial W^{k}}{\partial q^{i}}\right)=\epsilon,
\end{equation}
with $k$ number of real valued functions $W^\kappa$ on $Q$. Here, $\epsilon$ is a real constant. Let us exhibit this geometrically.

Assume that a section $\bf{X}$ is a HDW $k$-vector field on $T_k^*Q$ satisfying (\ref{HamEq-k}). 
Using a closed section $\boldsymbol{\gamma}$ of the $k$-cotangent bundle $T_k^*Q$, we are defining a section $\bf{X}^{\boldsymbol{\gamma}}$ of the $k$-tangent bundle $T^kQ$ as follows
\begin{equation} \label{X-Gamma-k}
\mathbf{X}^{\boldsymbol{\gamma}}:=T^k\boldsymbol{\pi}_Q\circ \mathbf{X} \circ \boldsymbol{\gamma}: Q \longrightarrow T^kQ. 
\end{equation}
Here, $T^k\boldsymbol{\pi}_Q$ is the $k$-th tangent mapping of the projection $\boldsymbol{\pi}_Q$, defined according to \eqref{k-lift}. Let us consider the decompositions of the $k$-vector fields $\mathbf{X}$ (defined on $T_k^*Q$) and $\mathbf{X}^{\boldsymbol{\gamma}}$ (defined on $Q$) given by 
\begin{equation}
\mathbf{X}=[X_\kappa]=(X_1,\dots,X_k), \qquad \mathbf{X}^{\boldsymbol{\gamma}}=[X_\kappa^{\boldsymbol{\gamma}}]=
(X^{\boldsymbol{\gamma}}_1,...,X^{\boldsymbol{\gamma}}_k),
\end{equation}
respectively. Notice that, in these decompositions, each $X_\kappa$ is a classical vector field on $T^*_kQ$ whereas each $X_\kappa^{\boldsymbol{\gamma}}$ is a classical vector field on $Q$. We have that 
 \begin{equation}
\mathbf{X}\circ\boldsymbol{\gamma}= [X_\kappa\circ \boldsymbol{\gamma}]=(X_1\circ\boldsymbol{\gamma},\dots,X_k\circ\boldsymbol{\gamma})
\end{equation}
By referring to the definition of $T^k\boldsymbol{\pi}_Q$ and in terms of the decompositions, the definition in \eqref{X-Gamma-k} turns out to be
\begin{equation} \label{X_k^g}
X_\kappa^{\boldsymbol{\gamma}}=T\boldsymbol{\pi}_Q \circ X_\kappa \circ \boldsymbol{\gamma}:Q\longrightarrow TQ
\end{equation}
for any $\kappa$ running from $1$ to $k$. We present a commutative diagram summarizing the discussions.
\begin{equation} \label{diagramX-Xgamma}
\xymatrix{ T^{*}_kQ
\ar[dd]^{\boldsymbol{\pi}_Q} \ar[rrrrr]^{\mathbf{X}=(X_1,\dots X_k)}&   &&& &T^kT_k^{*}Q\ar[dd]^{T^k\boldsymbol{\pi}_Q=(T\boldsymbol{\pi}_Q,\dots,T\boldsymbol{\pi}_Q)}\\
  &  & &\\
 Q\ar@/^2pc/[uu]^{\boldsymbol{\gamma}=(\gamma^1,\dots,\gamma^k)}\ar[rrrrr]_{\mathbf{X}^{\boldsymbol{\gamma}}=(X_1^{\boldsymbol{\gamma}},\dots,X_k^{\boldsymbol{\gamma}})}& && & & T^kQ \ar@/^2pc/[uu]^{{T^k\boldsymbol{\gamma}}=(T^k\gamma^1,\dots,T^k\gamma^k)}}
\end{equation} 
where $T^k\boldsymbol{\gamma}$ is the $k$-tangent lift of $\boldsymbol{\gamma}$.

As we previously indicated, each component $X_\kappa$ of $\bf{X}$ is a vector field on $T^*_kQ$. Consider Darboux' coordinates $(q^i,p_i^\kappa)$ on $T_k^*Q$, assume that each $X_\kappa$ has a local realization as defined in (\ref{coorxa}). In this frame, each component $X_\kappa^{\boldsymbol{\gamma}}$ of the section $\mathbf{X}^{\boldsymbol{\gamma}}$ defined in (\ref{X_k^g}) can locally be computed as
 \begin{equation}
  \label{coorxa-gamma}
   X_\kappa^{\boldsymbol{\gamma}} =(X_\kappa\circ \boldsymbol{\gamma})^i\frac{\partial}{\partial q^i}.
\end{equation}

The space $\mathrm{V}\boldsymbol{\pi}_Q$ of vertical vectors on the total space of the $k$-cotangent bundle $T^*_k Q$ are vectors in $TT^*_k Q$ lying in the kernel of the mapping $T\boldsymbol{\pi}_Q$. We have the following decomposition \cite{CaGrMaMaMuRo06} 
\begin{equation} \label{decomp}
T_{\boldsymbol{\gamma}(q)} T_k^*Q=T_{\boldsymbol{\gamma}(q)}(\operatorname{Im} \boldsymbol{\gamma}) \oplus \mathrm{V}_{\boldsymbol{\gamma}(q)} \boldsymbol{\pi}_{Q}.
\end{equation}
We generalize this decomposition to the $k$-tangent bundles as follows. Define the space of $k$-vertical vectors $\mathrm{V}^k\boldsymbol{\pi}_{Q}$ as a subbundle of the $k$-tangent bundle $T^k T_k^*Q$
\begin{equation}
\mathrm{V}^k\boldsymbol{\pi}_{Q}=\left\{\mathbf{V}\in T^k T_k^*Q: T^k \boldsymbol{\pi}_{Q}(\mathbf{V})=0  \right\}.
\end{equation}
Notice that we can write the space of $k$-vertical vectors $\mathrm{V}^k\boldsymbol{\pi}_{Q}$ as the Whitney sum of the space of vertical vectors $\mathrm{V}\boldsymbol{\pi}_Q$ that is
\begin{equation}
\mathrm{V}^k\boldsymbol{\pi}_{Q}=\mathrm{V}\boldsymbol{\pi}_Q\oplus_{T_k^*Q} \dots \oplus_{T_k^*Q}  \mathrm{V}\boldsymbol{\pi}_Q.
\end{equation}
Referring to this identification, we decompose the $k$-tangent vectors on the image space of $\boldsymbol{\gamma} $ as follows
\begin{equation} \label{decomp-k}
T_{\boldsymbol{\gamma}(q)}^k T_k^*Q=T^k_{\boldsymbol{\gamma}(q)}(\operatorname{Im} \gamma) \oplus \mathrm{V}^k_{\boldsymbol{\gamma}(q)} \boldsymbol{\pi}_{Q}.
\end{equation}
Let us first concentrate the following lemma which is useful in the proof of some forthcoming theorems. 

\begin{lemma} \label{commute}
Consider the canonically $k$-symplectic $k$-cotangent bundle $(T^*_kQ,[\Omega^\kappa],V)$.  
Let $\mathbf{X}=[X_\kappa]$ be a $k$-vector field on $T^*_kQ$ and $\mathbf{Y}=[Y_\kappa]$ be a $k$-vector field on $Q$. Assume also that $\mathbf{X}-T^k\gamma(\mathbf{Y})$ lies in the kernel of the mapping  $\flat$ in (\ref{flat}). Then the following equality holds
\begin{equation} \label{commute-eqn}
\boldsymbol{\gamma}^*\sum_{\kappa=1}^k \iota_{X_\kappa}\Omega^{\kappa} =\sum _{\kappa=1}^k\iota_{Y_\kappa}(\boldsymbol{\gamma}^*\Omega^{\kappa}),
\end{equation}
where $\boldsymbol{\gamma}$ is a section of $\boldsymbol{\pi}_Q$.  
\end{lemma}
\begin{proof}
Consider an arbitrary vector $v\in TT^*_kQ$ and compute the expression
\begin{equation} \label{gamma-rel-1}
\boldsymbol{\gamma}^*\sum_{\kappa=1}^k \iota_{X_{\kappa}}\Omega^{\kappa}(v)= \sum_{\kappa=1}^k \iota_{X_{\kappa}}\Omega^{\kappa}(T\gamma(v))= \sum _{\kappa=1}^k\Omega^{\kappa}(X_{\kappa},T\gamma(v))   
\end{equation}
and the following one
\begin{equation} \label{gamma-rel-2}
 \sum _{\kappa=1}^k\iota_{Y_\kappa}\boldsymbol{\gamma}^*\Omega^{\kappa}(v)=\sum_{\kappa=1}^k\boldsymbol{\gamma}^*\Omega^{\kappa}(Y_\kappa,v)= \sum_{\kappa=1}^k\Omega^{\kappa}(T\boldsymbol{\gamma}(Y_\kappa),T\gamma(v)). 
 \end{equation}
The expressions (\ref{gamma-rel-1}) and (\ref{gamma-rel-2}) coincide if and only if 
\begin{equation}
 \sum_{\kappa=1}^k\Omega^{\kappa}\Big(T\boldsymbol{\gamma}(Y_\kappa)-X_{\kappa},T\boldsymbol{\gamma}(v)\Big)=\flat(T^k\gamma(\mathbf{Y})-\mathbf{X})(T\boldsymbol{\gamma}(v))=0,
  \end{equation}
  for all vectors in form $T\boldsymbol{\gamma}(v)$. On the other hand $\mathbf{X}-T^k\gamma(\mathbf{Y})$ is a vertical vector field over the image space of $\boldsymbol{\gamma}$ with respect to the projection $\boldsymbol{\pi}_Q$ which reads that $\flat(\mathbf{X}-T^k\gamma(\mathbf{Y}))$ vanishes for all vertical tangent vectors with respect to $\boldsymbol{\pi}_Q$. 
Since the decomposition (\ref{decomp-k}), we deduce that \eqref{commute-eqn} holds if $\flat(\mathbf{X}-T^k\gamma(\mathbf{Y}))=0$. 
\end{proof}

Let us recall here a geometric version of the Hamilton-Jacobi formalism for $k$-symplectic framework, see \cite{LeDiMaSaVi10,LSV}. In the present work, we exhibit the proof in a coordinate free way.

\begin{theorem} \label{thm-k-symp-HJ}
Let $\bf{X}$ be an integrable solution of the HDW equation (\ref{generic0}) and $\boldsymbol{\gamma}$ be a closed section of $\boldsymbol{\pi}_Q$, then the following statements are equivalent:
\begin{itemize}
\item[1.] If $\sigma:U\subset \mathbb{R}^k\mapsto Q$ is an integral section of $\bf{X}^{\boldsymbol{\gamma}}$ then $\boldsymbol{\gamma}\circ\sigma$ is a solution $\bf{X}$.
\item[2.] $\mathbf{X} \circ \boldsymbol{\gamma}-T^{k} \boldsymbol{\gamma}(\mathbf{X}^{\boldsymbol{\gamma}}) \in\ker \flat$.
\item[3.] $d(H\circ \boldsymbol{\gamma})=0,$
\end{itemize}
where $T^{k} \boldsymbol{\gamma}$ is the $k$-tangent lift of $\boldsymbol{\gamma}$, whereas the operator $\flat$ is the one in \eqref{flat}.
\end{theorem}
\begin{proof} To prove this theorem we first show that $(1)$ if and only $(3)$ and $(2)$ if and only if $(3)$. 

\noindent $(1)\Longrightarrow(3):$ Let $\sigma$ be an integral curve of $\bf{X}^{\boldsymbol{\gamma}}$. The condition (1) gives that 
\begin{equation}
{\sigma}_{*}(\mathbf{t})\left(\frac{\displaystyle\partial}{\partial
t^{\kappa}}\right) = X^\gamma_{a}(\boldsymbol{\sigma} (\mathbf{t})) \quad  \Rightarrow \quad 
({\boldsymbol{\gamma}\circ\sigma)}_{*}(\mathbf{t})\left(\frac{\displaystyle\partial}{\partial
t^{\kappa}}\right) = X_{\kappa}( {\boldsymbol{\gamma}\circ\sigma} (\mathbf{t}))
\end{equation}
Applying the push forward $\boldsymbol{\gamma}_*$ to the former equation we see that $X^\gamma_a$ and $X_{\kappa}$ are $\boldsymbol{\gamma}$-related over the solution hypersurfaces. So, we have
\begin{equation}
\d (H\circ\boldsymbol{\gamma})=\boldsymbol{\gamma}^*\d H=\boldsymbol{\gamma}^*\sum_{\kappa=1}^k \iota_{{X_{\kappa}}}\Omega^{\kappa}=\sum_{\kappa=1}^k i_{{X_{\kappa}^{\gamma}}}(\gamma^{\kappa})^*\Omega_Q=0.
\end{equation}

\noindent $(3)\Longrightarrow(1):$ Let us assume that $\d (H\circ\boldsymbol{\gamma})=0$ and 
$\sigma:U\subset \mathbb{R}^k\mapsto Q$ is an integral section of $\bf{X}^{\boldsymbol{\gamma}}$. We have to show that
$\boldsymbol{\gamma}\circ\sigma$ is a solution of $\bf{X}$ (a solution of Hamilton-De Donder-Weyl equation). Notice that $\boldsymbol{\gamma}\circ\sigma$ is an integral section of the $k$-symplectic vector field $\bf{X}$ defined on $\sT^*_kQ$. From Theorem \eqref{thm-k-symp-HJ},  $\bf X$ and $\mathbf{X}^{\boldsymbol{\gamma}}$ are $\gamma$-related up to the kernel of the $\flat$ map. We have to show that 
$$ \flat({\bf X})=\d H  $$
on the image of $\sigma$. From (3) we have
\begin{equation}
\begin{split}
0&=\d (H\circ\boldsymbol{\gamma})(\sigma(t))=\boldsymbol{\gamma}^*\d H(\sigma(t))=\sum_{\kappa=1}^k \iota_{{X_{\kappa}^{\gamma}}(\sigma(t))}(\boldsymbol{\gamma})^*\Omega^{\kappa}(\sigma(t))\\&=\boldsymbol{\gamma}^*\left(\sum_{\kappa=1}^k \iota_{{X_{\kappa}}}\Omega^{\kappa}(\sigma(t))\right)=\boldsymbol{\gamma}^* \flat({\bf X})(\sigma(t)),
\end{split}
\end{equation}
where the third equality comes from (\ref{zero}). Hence $\boldsymbol{\gamma}\circ\sigma$ is a solution of ${\bf X}$.

\noindent $(2)\Longrightarrow(3):$ Assume that $(2)$ holds, that is $\mathbf{X}-T^k\boldsymbol{\gamma}(\mathbf{X}^{\boldsymbol{\gamma}})$ belongs to $\ker\flat$. Then we have 
\begin{equation} \label{1-to-2}
\d (H\circ\boldsymbol{\gamma})=\boldsymbol{\gamma}^*\d H=\boldsymbol{\gamma}^*\sum_{\kappa=1}^k \iota_{{X_{\kappa}}}\Omega^{\kappa}=\sum_{\kappa=1}^k \iota_{{X_{\kappa}^{\gamma}}}\boldsymbol{\gamma}^*\Omega_Q=0  
\end{equation}
where we have used Lemma (\ref{commute}). However, it is (\ref{zero}) which provides the last equality.

\noindent $(3)\Longrightarrow (2):$ Assume that $(3)$ holds. Define a $k$-vector field $\mathbf{D}=\mathbf{X}-T^{k} \boldsymbol{\gamma} \left(\mathbf{X}^{\gamma}\right)$ on $T_k^*Q$. If the vector field $\mathbf{D}$ belongs to $\ker\flat$, we have finished the proof. For this, first see that $\mathbf{D}$ is a vertical vector field with respect to the projection $\boldsymbol{\pi}_Q$, that is
\begin{equation} \label{Disvert}
\begin{split}
T^{k}\boldsymbol{\pi}_Q\circ \mathbf{D} &=T^k\boldsymbol{\pi}_Q \Big(\mathbf{X}-T^{k}\boldsymbol{\gamma} \left(\mathbf{X}^{\gamma}\right)\Big)\\
&
=T^{k}\boldsymbol{\pi}_Q  \circ \mathbf{X} - T^{k}\pi_Q \circ T^{k} \boldsymbol{\gamma}\left(\mathbf{X}^{\gamma}\right)
\\
&=T^{k}\boldsymbol{\pi}_Q  \circ \mathbf{X} -T^{k}\boldsymbol{\pi}_Q \circ T^{k}  \boldsymbol{\gamma}\circ T^{k}\boldsymbol{\pi}_Q \circ 
\mathbf{X}  
\\
&=T^{k}\boldsymbol{\pi}_Q  \circ \mathbf{X}  -T^{k}\boldsymbol{\pi}_Q \circ 
\mathbf{X}  =0,
\end{split}
\end{equation}
where we have used that $T^{k}  \boldsymbol{\gamma}\circ T^{k}\boldsymbol{\pi}_Q$ is equal to the $k$-tangent lift of the identity mapping $\boldsymbol{\pi}_Q \circ  \boldsymbol{\gamma}$.  
Now, see that $\mathbf{D}$ vanishes identically on the image space of $\gamma$ for the vector fields on $T^*_kQ$ in the form $\boldsymbol{\gamma}_*Y$ where $Y$ is any vector field on $Q$. 
\begin{eqnarray*}
\flat(\mathbf{D})(\boldsymbol{\gamma}_*Y)&=&\sum_{\kappa=1}^k  i_{{D_a}}\Omega^{\kappa}(\boldsymbol{\gamma}_*Y)=\sum_{\kappa=1}^k  \Omega^{\kappa}(X_{\kappa}-\boldsymbol{\gamma}_*X_{\kappa}^\gamma,\boldsymbol{\gamma}_*Y)\\&=&\sum_{\kappa=1}^k  \Omega^{\kappa}(X_{\kappa},\boldsymbol{\gamma}_*Y)-\sum_{\kappa=1}^k  \Omega^{\kappa}(\boldsymbol{\gamma}_*X_{\kappa}^\gamma,\boldsymbol{\gamma}_*Y)\\
&=& \sum_{\kappa=1}^k \iota_{X_{\kappa}}\Omega^{\kappa}(\boldsymbol{\gamma}_*Y)-\sum_{\kappa=1}^k\boldsymbol{\gamma}^*(\Omega^{\kappa})(X_{\kappa}^\gamma,Y) \\ &=&\boldsymbol{\gamma}^*\sum_{\kappa=1}^k(\iota_{X_{\kappa}}\Omega^{\kappa})(Y) -\sum_{\kappa=1}^k\boldsymbol{\gamma}^*(\Omega^{\kappa})(X_{\kappa}^\gamma,Y) \\
&=&(\boldsymbol{\gamma}^*\d H)(Y) -\sum_{\kappa=1}^k\boldsymbol{\gamma}^*(\Omega^{\kappa})(X_{\kappa}^\gamma,Y)=0
\end{eqnarray*}
where we have employed condition (2) in the first term and the closedness of $\gamma$ in the second term given in the last line. It follows that $\mathbf{D}$ is in the kernel of $\flat$ that implies condition (1). 

\end{proof}

\section{Locally conformal k-symplectic framework} \label{Sec-lcks}
  
\subsection{Locally conformal k-symplectic manifolds} \label{mot-lcks}

\noindent \textbf{Motivation.} In \cite{vaisman}, one can review a first discussion on the necessity of locally conformal symplectic manifolds to study Hamiltonian dynamics. In this section, we provide a similar discussion justifying the necessity of what we call locally conformal $k$-symplectic manifolds while studying HDW dynamics on the $k$-symplectic framework. We start working on a manifold $M$ with an open covering $\{U_\alpha\}$. Assume also that on each open chart $U_\alpha$, there exist $k$ number of closed two-forms $[\Omega^{\kappa}_\alpha]=(\Omega^{1}_\alpha,\dots, \Omega^{k}_\alpha)$ satisfying the conditions in (\ref{condksymp}) for an integrable distribution $V_\alpha$. In other words, we impose that the family $(U_\alpha,[\Omega^{\kappa}_\alpha],V_\alpha)$ is a $k$-symplectic manifold. For a local Hamiltonian function $H_\alpha$ on this local $k$-symplectic space, the local HDW equations can be written as
\begin{equation} \label{Loc-k-Ham-Eq}
\sum_{\kappa=1}^k\iota_{(X_\alpha)_\kappa}\Omega_\alpha^\kappa=d H_\alpha.
\end{equation}
The solutions $(X_\alpha)_\kappa$ of this covariant equation are vector fields on $U_\alpha$, and 
the collection of these vector fields $(X_\alpha)_\kappa$ determines a $k$-vector field on $U_\alpha$ taking values in the $k$-tangent bundle $T^kU_\alpha$.  We denote this local $k$-vector field by $\mathbf{X}_\alpha:=[(X_\alpha)_\kappa]$. 
\bigskip

\noindent 
We wish to glue the local $k$-vector fields $\mathbf{X}_\alpha$ determined in \eqref{Loc-k-Ham-Eq} up to a global $k$-vector field $\mathbf{X}$ on $M$ in order to arrive at global dynamics on $M$. In this case, as it is well-know, the allowed coordinate transformations are the ones preserving the family of two-forms $[\Omega^{\kappa}_\alpha]$. Notice that, multiplying the equality \eqref{Loc-k-Ham-Eq} by a scalar $\lambda_{\beta\alpha}$, we can preserve the structure on overlapping charts $U_\alpha$ and $U_\beta$. That is, the local $k$-vector fields coincide
\begin{equation}
(X_{\kappa})_\alpha=(X_{\kappa})_\beta,
\end{equation}
if the closed two-forms $[\Omega^{\kappa}_\alpha]$ and the Hamiltonian function $H_\alpha$ obey conformal transformations 
\begin{equation} \label{Lambda}
\Omega^{\kappa}_\beta=\lambda_{\beta\alpha}\Omega^{\kappa}_\alpha, \quad\forall \kappa=1,\dots,k, \qquad H_\beta=\lambda_{\beta\alpha} H_\alpha,
\end{equation} 
respectively. This can also be recognized by applying these conformal relations directly to the local equations of motion in \eqref{he20}. It is immediate to realize that the scalars must satisfy the cocycle condition
\begin{equation} \label{cocyle}
\lambda_{\delta\beta}\lambda_{\beta\alpha} = \lambda_{\delta\alpha}.
\end{equation}
This leads us to define a real line bundle $L \mapsto M$ over the base manifold $M$. This way, one can glue up the local two-forms $\Omega_\alpha^\kappa$ to the line bundle valued two-form $\tilde{\Omega}^\kappa$ on $M$. Further, we argue that 
the family $[\Omega_\alpha^\kappa]$ of closed one-forms on $U_\alpha$ determines a family of line bundle valued two-forms $[\tilde{\Omega}^\kappa]$ on $M$.
\bigskip

\noindent 
The cocycle condition (\ref{cocyle}) manifests the existence of a family of local functions $\sigma_\alpha$ on each $U_\alpha$ satisfying the relation 
\begin{equation} \label{transition}
\lambda_{\beta \alpha}=e^{\sigma_{\alpha}}/e^{\sigma_{\beta}}=e^{-(\sigma_{\beta}-\sigma_{\alpha})}
\end{equation}
in overlapping charts. The substitution of this realization into the first equation in (\ref{Lambda}) motivates us to define local two-forms 
\begin{equation} \label{ua}
\Omega^{\kappa}\vert_\alpha:=e^{\sigma_\alpha}\Omega^{\kappa}_\alpha.
\end{equation}
Since, these local forms give $\Omega^{\kappa}\vert_\alpha=\Omega^{\kappa}\vert_\beta$ on the intersections $U_\alpha\cap U_\beta $, they can be glued to a real valued two-form $\Omega^{\kappa}$ on $M$. So, the local family $[\Omega^{\kappa}\vert_\alpha]$ determines a global real valued two-form family $[\Omega^{\kappa}]$. To sum up, we say that there are two global two-form families, namely $[\Omega^\kappa]$ (real valued) and $[\tilde{\Omega}^\kappa]$ (line bundle valued), on $M$ with local realizations $[\Omega^\kappa\vert_\alpha]$ and $[\Omega_\alpha]$, respectively. These local two-forms are related as given in (\ref{ua}).  
\bigskip

\noindent In the light of the present discussion, we can argue that if one only imposes the invariance of the HDW equations when we are glueing up the local $k$-symplectic structures, the result is not necessarily a global $k$-symplectic structure. Instead, we arrive at what we call locally conformal $k$-symplectic structures. Here we present the formal  definition.
\bigskip

\noindent \textbf{The formal definition.}
Consider a $n+kn$
dimensional manifold $M$ with $k$ number of two-forms $\Omega^{\kappa}$, and $nk$-dimensional
integrable distribution $V$. $M$ is called a locally conformal k-symplectic (l.c.k-s.) manifold if $M$ has an open cover $\{U_\alpha\}$ and there exists a family of smooth functions $\sigma_\alpha$ on each chart $U_\alpha$ providing that
\begin{equation}\label{condksymp-loc}
(i)\quad \d(e^{-\sigma_\alpha}\Omega^{\kappa}\vert_\alpha)=0, \qquad
(ii) \quad \bigcap_{\kappa=1}^{k} \ker\Omega^{\kappa}=\{0\},\qquad
(iii) \quad \Omega^{\kappa}\vert_{ V\times V}=0.
\end{equation}
Here, $\Omega^{\kappa}\vert_\alpha$ denotes the restriction of the two-form $\Omega^{\kappa}$ to the chart $U_\alpha$. So, the local two-forms defined by $\Omega^{\kappa}_\alpha:=e^{-\sigma_\alpha}\Omega^{\kappa}\vert_\alpha$ are closed two-forms on $U_\alpha$. Note that the first condition in (\ref{condksymp-loc}) is local and the other two are global. By glueing up the local condition, we arrive at the following definition of the l.c.k-s. manifolds in terms of a Lee form. 
\bigskip

\noindent
Let $M$ be a manifold of dimension $n+kn$, then $(M,[\Omega^\kappa], \theta, V)$ is called a l.c.k-s. manifold if $\theta$ is a closed one-form on $M$, $V$ is an integrable $nk$-dimensional
distribution on $M$, and each $\Omega^{\kappa}$
is a $2$-form on $M$ satisfying
\begin{equation}\label{condksymp-glo}
(i)\quad \d\Omega^{\kappa}=\theta\wedge\Omega^{\kappa}, \qquad
(ii) \quad \bigcap_{\kappa=1}^{k} \ker\Omega^{\kappa}=\{0\},\qquad
(iii) \quad \Omega^{\kappa}\vert_{ V\times V}=0.
\end{equation}
The first condition in (\ref{condksymp-glo}) is obtained by glueing up the local one-forms, i.e., on each chart, $\theta\vert_\alpha=d\sigma_\alpha$. 
Here, we call the family $([\Omega^\kappa], \theta)$ of differential forms a l.c.k-s. structure. 
Let us discuss the first condition \eqref{condksymp-glo} in the realm of the Lichnerowicz-de Rham differential presented in (\ref{LdR}). We state that a family $(M,[\Omega^\kappa], \theta, V)$ satisfying the second and the third conditions in (\ref{condksymp-glo}) is a l.c.k-s. manifold if and only if the two-forms $[\Omega^{\kappa}]$ are lying in the kernel of the  Lichnerowicz-de Rham differential $d_\theta$ for the closed one-form $\theta$ on $M$. A l.c.k-s. manifold is said to be globally conformal $k$-symplectic manifold if the Lee form is an exact one-form. 
\bigskip

\noindent \textbf{Exact l.c.k-s. structures.}  We say that a l.c.k-s. structure $([\Omega^{\kappa}],\theta)$ is exact if $\Omega^{\kappa}=d_\theta \Upsilon^\kappa$ for a family of differential one-forms $[\Upsilon^\kappa]=(\Upsilon^1,\dots, \Upsilon^k)$ on $M$. We call such a set $[\Upsilon^\kappa]$ as Liouville one-form family for the l.c.k-s. structure. We denote an exact l.c.k-s. manifold by $(M,[d_\theta \Upsilon^\kappa],\theta)$ generated by the family  $[\Upsilon^\kappa]$. In this case, we define a set of vector fields $Z_{\Upsilon^\kappa}$ by
\begin{equation} \label{Louville-v-f}
\iota_{Z_{\Upsilon^\kappa}}\Omega^\kappa=\Upsilon^\kappa, \qquad \Omega^\kappa=d_\theta\Upsilon^\kappa.
\end{equation}
Notice that this family determines a $k$-vector field $\mathbf{Z}=[Z_{\Upsilon^\kappa}]$ on $M$, which we call Liouville $k$-vector field. Applying $\iota_{Z_{\Upsilon^\kappa}}$ to the both hand side of the first equation in (\ref{Louville-v-f}) we see that $\iota_{Z_{\lambda^\kappa}}\Upsilon^\kappa=0$ where there is no sum on $\kappa$. 
\bigskip

\noindent \textbf{Lagrangian submanifolds of l.c.k-s. manifolds.} We define Lagrangian submanifolds of l.c.k-s. manifolds similar to the Lagrangian submanifolds of $k$-symplectic manifolds. Let  $N$ be a submanifold of a l.c.k-s. manifold $(M,[\Omega^{\kappa}],V)$. Then l.c.k-s. dual of the tangent space $T_zN$ is defined as the same with \eqref{sdual} whereas in this time the family $[\Omega^{\kappa}]$ does not consist of closed two-forms but the ones satisfying (\ref{condksymp-loc}). A submanifold $N$ is called isotropic if $T_zN\subset(T_zN)^\perp$ for all $z$. $N$ is called coisotropic if $(T_zN)^\perp\subset T_zN$ whereas it is called  Lagrangian if $(T_zN)^\perp= T_zN$, for all $z$. Notice that, the definition of Lagrangian submanifold only related with the second and the third conditions in (\ref{condksymp-loc}). These are the same with the ones in \eqref{condksymp}. So that one may expect that Lagrangian submanifolds of l.c.k-s. manifolds and $k$-symplectic manifolds be in an harmony. For example, as in the case of $k$-symplectic geometry, the submanifold of  $M$ integrating the distribution $V$ is also a Lagrangian submanifold of l.c.k-s. manifold. 

\subsection{L.c.k-s. structures on the Whitney sum of cotangent bundles} \label{Wslcs}

Let $Q$ be a manifold and $\vartheta$ be a closed form on it. Pull $\vartheta$ back to the cotangent bundle $T^*Q$ by means of the cotangent bundle projection $\pi_Q$ in order to define a semi-basic and closed one-form $\theta$. Recall the l.c.s. manifold $T^*_\theta Q$ determined in \eqref{T*_Q}. Define the Whitney sum of $k$ number of such bundles
\begin{equation} \label{lcks-T*Q}
T^*_{k,\theta}Q:=T^*_\theta Q\oplus_Q T^*_\theta Q \oplus_Q \dots \oplus_Q T^*_\theta Q.
\end{equation}
This product is topologically the same with the one in (\ref{Tk*Q}), and the subscript $\theta$ is placed just to remind the Lee form. 
\bigskip 

\noindent 
There exist $k$ number of projections $\pi^{\kappa}$ from $T^*_{k,\theta}Q$ to the $\kappa$-th copy of $T^*_\theta Q$ in the sum. These projections satisfy the commutative diagram in (\ref{pi-a}). Using them, we introduce $k$ number of differential two-forms 
 \begin{equation}
  \Omega^{\kappa}_\theta:=(\pi^{\kappa})^*\Omega_\theta 
\label{symforms}
 \end{equation}
on $T^*_{k,\theta}Q$. Here, $\Omega_\theta$ is the almost symplectic two-form defined in (\ref{omega_theta}). Another way to define these two-forms is as follows. Consider the closed one-form $\vartheta$ defining the Lee-form $\theta=\pi_Q^*\vartheta$ on $T^*_\theta Q$. It can be pulled back to $T^*_{k,\theta}Q$ by means of the $k$-th order cotangent bundle projection $\boldsymbol{\pi}_Q$. In this case, we arrive at a semi-basic one-form section on $T^*_{k,\theta}Q$. Referring to the commutative diagram \eqref{pi-a}, this one-form coincides with the pull-back of the Lee-form $\theta$ on $T^*_\theta Q$ by the projection $\pi^{\kappa}$, that is 
\begin{equation}
\boldsymbol{\pi}_Q^*(\vartheta)=(\pi^{\kappa})^*(\pi_Q)^*(\vartheta)=(\pi^{\kappa})^*\theta.
\end{equation}
Let us make abuse of notation we denote $\boldsymbol{\pi}_Q^*(\vartheta)$ by $\theta$ as well. 
Notice that, exterior derivatives of the two-forms $(\Omega_\theta^\kappa)$ defined in \eqref{symforms} are closed up to the Lee-form $\theta$, that is  
$$\d\Omega^{\kappa}_\theta=\d(\pi^{\kappa})^*\Omega_\theta=(\pi^{\kappa})^*\d\Omega_\theta=(\pi^{\kappa})^*(\theta\wedge\Omega_\theta)=
(\pi^{\kappa})^*\theta\wedge(\pi^{\kappa})^*\Omega_\theta=\theta\wedge\Omega^{\kappa}. $$
\bigskip

\noindent On the base manifold $Q$, consider a local coordinate system $(q^i)$. In this frame, the closed one-form  is written as $\vartheta=\vartheta_i dq^i$. Accordingly, on the Darboux' coordinates $(q^i,p^\kappa_i)$ on $T_{k,\theta}^*Q$, we have the following local realizations of the almost symplectic two-forms 
\begin{eqnarray}\label{Omega}
{\Omega^{\kappa}_\theta}=  \d p^\kappa_i\wedge\d q^i+\vartheta_i(q)p_j^\kappa\d q^i\wedge\d q^j.
\end{eqnarray}
One can see that the kernel of $\Omega^{\kappa}_\theta$ is locally generated by the $k-1$ local vector fields $\langle {\partial}/{\partial p^\lambda_i}\rangle$ such that $\lambda\neq \kappa$. Therefore, the intersection of the kernels of almost symplectic two-forms is trivial, and the second condition in (\ref{condksymp-glo}) is satisfied. Considering once more the vertical distribution $V$ exhibited in \eqref{k-symp-s-loc} 
we have that the third condition in (\ref{condksymp-glo}) is fullfilled. As a result we arrive at that the family $(T_{k,\theta}^*Q,[\Omega_\theta^\kappa],\theta,V)$ is a l.c.k-s. manifold.
\bigskip 

\noindent \textbf{$d_\vartheta$ closed Section of l.c.k-s. manifold $T_{k,\theta}^*Q$.}
Let us determine horizontal Lagrangian submanifolds of l.c.k-s. manifold $T_{k,\theta}^*Q$. For this we consider a section $\boldsymbol{\gamma}=[\gamma^\kappa]$ of the bundle $T_{k,\theta}^*Q$. We perform the  following calculation
\begin{equation}  \label{LdR-Lag}
\begin{split}
(\gamma^\kappa)^*\Omega_\theta^\kappa&=(\gamma^\kappa)^*(-d\Theta^\kappa_Q+\theta\wedge \Theta^\kappa_Q)
=-d(\gamma^\kappa)^*(\Theta^\kappa_Q)+(\gamma^\kappa)^*(\theta\wedge \Theta^\kappa_Q)
\\ &= -d\gamma^\kappa+(\gamma^\kappa)^*\theta\wedge (\gamma^\kappa)^*\Theta_Q
=-d\gamma^\kappa+\vartheta\wedge \gamma^\kappa=-d_\vartheta\gamma^\kappa
\end{split}
 \end{equation}
where we used the identities $(\gamma^\kappa)^*\Theta^\kappa_Q=(\gamma^\kappa)$ and $(\gamma^\kappa)^*\theta=\vartheta$. This gives that the image space of $\boldsymbol{\gamma}=[\gamma^\kappa]$ is an isotropic submanifold of $T_\theta^*Q$ if and only if $d_\vartheta \gamma^{\kappa}=0$ for all $\kappa$ from $1$ to $k$. We have discussed that the Lagrangian submanifolds of $k$-symplectic and l.c.k-s. manifolds are in harmony. For an another example of this, see that each two-form in $[\Omega^\kappa_\theta]$ is the sum of the canonical two-form and the wedge product of two semi-basic one-forms. This is given in \eqref{omega_theta}. So that, each $\Omega^\kappa_\theta$ vanishes on vector fields those parallel to the projection $\boldsymbol{\pi}_Q$. That is, fibers of the projection are isotropic submanifolds of l.c.k-s. manifold $T_{k,\theta}^*Q$. Another way to see this directly refer to the local realization in \eqref{Omega}. Since the tangent space to image of $\boldsymbol{\gamma}$ at each point is complement those tangent to the fibers, we argue that both the image of  $\boldsymbol{\gamma}$ and the fibers of the projection are Lagrangian submanifolds. See \cite{LV} for a similar discussions done in $k$-symplectic geometry. Since $d_\vartheta^2$ is identically zero, the image space of the one-form $d_\vartheta W$ is a Lagrangian submanifold of $T^*_\theta Q$ for a function $W$ defined on $Q$. So that, for a family of functions $(W^1,\dots, W^k)$ the set
\begin{equation}
\text{im}(d_\vartheta W^1)\oplus_Q \dots  \oplus_Q \text{im}(d_\vartheta W^k) \subset T^*Q\oplus_Q\dots  \oplus_Q T^*Q
 \end{equation} 
is a Lagrangian submanifold of $T_{k,\theta}^*Q$. 

\subsection{Dynamics in l.c.k-s. manifolds} \label{dlcks}

Recall the discussion done in the beginning of Subsection \ref{mot-lcks} regarding the globalization of the local HDW equations \eqref{Loc-k-Ham-Eq}. We have exhibited that the two-forms $\Omega^{\kappa}$ in the global picture take the local realization $\Omega^{\kappa}\vert_\alpha=e^{\sigma_\alpha}\Omega^{\kappa}_\alpha$ on a local chart $U_\alpha$. On the other hand, the collection of local Hamiltonian functions $H_\alpha$ generating the local HDW equations \eqref{Loc-k-Ham-Eq} determine a line bundle valued Hamiltonian (a "twisted Hamiltonian") function $\tilde{H}$. On the other hand, the local functions $H\vert_\alpha$, defined by 
\begin{equation} \label{glueHamFunc}
 H\vert_\alpha=e^{\sigma_{\alpha}}H_\alpha,
 \end{equation}
give a real valued Hamiltonian function $H$ on the manifold $M$. That is we have two global functions on $M$, the line bundle valued Hamiltonian function $\tilde{H}$ and real valued Hamiltonian function ${H}$. On a chart $U_\alpha$, these two functions reduce to $h_\alpha$ and $h\vert_\alpha$, respectively, and they satisfy the relation \eqref{glueHamFunc}.
\bigskip

\noindent We are now ready to glue up the local HDW equations \eqref{Loc-k-Ham-Eq}. Consider a l.c.k-s. manifold $(M,[\Omega^\kappa], \theta, V)$, and a Hamiltonian function $H$ on $M$. DWH equation in this framework is 
\begin{equation}\label{global}
\flat(\mathbf{X})=d_\theta H, 
\end{equation}
where $\d_\theta$ is the Lichnerowicz-deRham differential defined on $M$. Here, the mapping $\flat$ is defined as in (\ref{flat}).
In other words, we have
\begin{equation}\label{Sum2}
\sum_{\kappa=1}^k \iota_{X_{\kappa}}\Omega^{\kappa}=d_\theta H.
\end{equation}

\subsection{Hamilton-Jacobi theory} \label{HJ-lcks}

We will formulate now a generalization of the Hamilton-Jacobi theory for $k$-symplectic manifolds to the l.c.k-s. case. For this, consider the locally conformal $k$-symplectic manifold $T_{k,\theta}^*Q$ presented in \eqref{lcks-T*Q}. Consider a section $\boldsymbol{\gamma}$ of the fibration $\boldsymbol{\pi}_Q$.  This section can be written as a combination of $k$ number of differential one-forms $\boldsymbol{\gamma}=(\gamma^1,...,\gamma^k)$ as defined in \eqref{Gamma-gamma-a}. We wish that each $\gamma^{\kappa}$ is closed with respect to the almost symplectic two-form $\Omega^{\kappa}_\theta$ defind in (\ref{omega_theta}). In the light of the calculation done in (\ref{LdR-Lag}), this is equivalent to the closure of the sections with respect to the Lichnerowicz-deRham differential $d_\vartheta$ where $\vartheta$ is the closed one-form generating the Lee form $\theta=\pi_Q(\vartheta)$. We denote a $k$-vector field on $T^*_{k,\theta}Q$ by $\mathbf{Z}$. As in (\ref{X-Gamma-k}), using $\boldsymbol{\gamma}$, we can construct a $k$-vector field $\mathbf{Z}^\gamma$ on $Q$ by
$$ \mathbf{Z}^{\boldsymbol{\gamma}}=T^k\boldsymbol{\pi}_Q\circ \mathbf{Z} \circ \boldsymbol{\gamma},$$
where the mapping $T^k\boldsymbol{\pi}_Q$ is $k$-tangent lift of the projection and it is defined as in (\ref{diagramm}). 
Notice that, the k-vector fields $\mathbf{Z}$ and $\mathbf{Z}^\gamma$ fulfill the commutativity diagram \eqref{diagramX-Xgamma}.

\begin{theorem} \label{HJ-l.c.k-s.-1}
Let $\mathbf{Z}$ be an integrable solution of HDW equation (\ref{Sum2}) and $\boldsymbol{\gamma}$ is a closed section. Then the followings statements are equivalent:
\begin{itemize}
\item[1.] If $\sigma:U\subset \mathbb{R}^k\mapsto Q$ is an integral section of $\mathbf{Z}^\gamma$ then $\gamma\circ\sigma$ is a solution of $\mathbf{Z}$.
\item[2.] $\mathbf{Z} \circ \boldsymbol{\gamma}-T^{k} \boldsymbol{\gamma} \left(\mathbf{Z}^{\boldsymbol{\gamma}}\right) \in\ker \flat$.
\item[3.] $d_\vartheta(H\circ \boldsymbol{\gamma})=0.$
\end{itemize}
\end{theorem}
\begin{proof} While proving this theorem we first show that $(1)$ and $(3)$ are in if and only if relationship then we show that $(2)$ and $(3)$ are in if and only if relationship. 

\noindent $(1)\Longrightarrow(3):$ Let $\sigma$ be an integral curve of $\mathbf{Z}^{\boldsymbol{\gamma}}$. The condition (1) gives that 
\begin{equation}
{\sigma}_{*}(\mathbf{t})\left(\frac{\displaystyle\partial}{\partial
t^{\kappa}}\right) = X^\gamma_{\kappa}(\boldsymbol{\sigma} (\mathbf{t})) \quad  \Rightarrow \quad 
({\boldsymbol{\gamma}\circ\sigma)}_{*}(\mathbf{t})\left(\frac{\displaystyle\partial}{\partial
t^{\kappa}}\right) = X_{\kappa}( {\boldsymbol{\gamma}\circ\sigma} (\mathbf{t}))
\end{equation}
Applying push forward $\boldsymbol{\gamma}_*$ to the former equation we see that $Z^\gamma_a$ and $Z_\kappa$ are $\boldsymbol{\gamma}$-related over the solution hypersurfaces. So that we have
\begin{equation}
\d_\vartheta(H\circ\boldsymbol{\gamma})=\boldsymbol{\gamma}^*\d_\theta H=\boldsymbol{\gamma}^*\sum_{\kappa=1}^k \iota_{{X_{\kappa}}}\Omega_\theta^\kappa=\sum_{\kappa=1}^k i_{{X_{\kappa}^{\gamma}}}(\gamma^{\kappa})^*\Omega_Q=0.
\end{equation}

\noindent $(3)\Longrightarrow(1):$ Let us assume that $\d_\vartheta(H\circ\boldsymbol{\gamma})=0$ and 
$\sigma:U\subset \mathbb{R}^k\mapsto Q$ is an integral section of $\bf{Z}^{\boldsymbol{\gamma}}$. We have to show that
$\boldsymbol{\gamma}\circ\sigma$ is a solution of $\bf{Z}$ (a solution of Hamilton-De Donder-Weyl equation). Let us notice that $\boldsymbol{\gamma}\circ\sigma$ is an integral section of the $k$-symplectic vector field $\bf{Z}$ defined on $T^*_{k,\theta}Q$. Referring to Theorem \eqref{HJ-l.c.k-s.-1},  $\bf Z$ and $\mathbf{Z}^{\boldsymbol{\gamma}}$ are $\gamma$-related up to kernel of $\flat$. We have to show that 
$$ \flat({\bf Z})=\d_\vartheta H  $$
on the image of $\sigma$. From (3) we have
\begin{equation}
\begin{split}
0&=\d_\vartheta(H\circ\boldsymbol{\gamma})(\sigma(t))=\boldsymbol{\gamma}^*\d_\theta H(\sigma(t))=\sum_{\kappa=1}^k \iota_{{X_{\kappa}^{\gamma}}(\sigma(t))}(\boldsymbol{\gamma})^*\Omega_\theta^\kappa(\sigma(t))\\&=\boldsymbol{\gamma}^*\sum_{\kappa=1}^k \iota_{{X_{\kappa}}}\Omega_\theta^\kappa(\sigma(t))=\boldsymbol{\gamma}^* \flat({\bf Z})(\sigma(t)),
\end{split}
\end{equation}
where the third equality comes from (\ref{zero}). Hence $\boldsymbol{\gamma}\circ\sigma$ is a solution of ${\bf Z}$.

\noindent $(2)\Longrightarrow (3):$ Notice that the second condition (3) can be rewritten as $\boldsymbol{\gamma}^*(d_\theta H)=0$ according to the following calculation 
\begin{equation} \label{Calc-d-Ld}
\begin{split}
d_\vartheta(H\circ \boldsymbol{\gamma})&=d(H\circ \boldsymbol{\gamma})-(H\circ \boldsymbol{\gamma})\vartheta
=d(\boldsymbol{\gamma}^* H)-\boldsymbol{\gamma}^*H\vartheta
\\
&=\boldsymbol{\gamma}^*dH-\boldsymbol{\gamma}^*H\boldsymbol{\gamma}^*\theta=\boldsymbol{\gamma}^*(dH-H\theta)
=\boldsymbol{\gamma}^*(d_\theta H),
\end{split}
\end{equation}
where we used the identity 
\begin{equation}
\boldsymbol{\gamma}^*\theta=\boldsymbol{\gamma}^*\boldsymbol{\pi}_Q^*\vartheta= (\boldsymbol{\pi}_Q \circ \boldsymbol{\gamma})^*\vartheta
=\vartheta
\end{equation}
in the second line of the calculation (\ref{Calc-d-Ld}). 
Assume now that the first condition holds, then Lemma (\ref{commute}) leads us to compute
\begin{equation}
\begin{split}
\boldsymbol{\gamma}^*(d_\theta H)&=\boldsymbol{\gamma}^*\sum_{\kappa=1}^k\iota_{Z_\kappa}\Omega^{\kappa}_\theta=\sum_{\kappa=1}^k\iota_{Z_\kappa^\gamma}\boldsymbol{\gamma}^*\Omega^{\kappa}_\theta \\ &=\sum_{\kappa=1}^k\iota_{Z_\kappa^\gamma}\boldsymbol{\gamma}^*(\pi^{\kappa})^*\Omega_\theta
=
\sum_{\kappa=1}^k\iota_{Z_\kappa^\gamma}( \pi^{\kappa} \circ \boldsymbol{\gamma})^*\Omega_\theta\\ &=\sum_{\kappa=1}^k\iota_{Z_\kappa^\gamma}( \gamma^{\kappa})^*\Omega_\theta= 0,
\end{split}
\end{equation}
where the last equality comes from the calculation in (\ref{LdR-Lag}) stating that $\boldsymbol{\gamma}^*\Omega^{\kappa}_\theta=0$ equivalent to $d_\theta \boldsymbol{\gamma}=0$.
\bigskip

\noindent $(3)\Longrightarrow (2):$ To prove that, same as in $k$-symplectic case, we define a $k$-vector field $\mathbf{D}=\mathbf{Z}-T^{k} \boldsymbol{\gamma}\left(\mathbf{Z}^{\boldsymbol{\gamma}}\right)$ on $T_{k,\theta}^*Q$. Our goal is to show that $\mathbf{D}$ belongs to $\ker\flat$ when it is restricted on the image space of $\boldsymbol{\gamma}$. For this end,  see that $\mathbf{D}$ is a vertical $k$-vector field over the image space of $\boldsymbol{\gamma}$ with respect to the projection $\boldsymbol{\pi}_Q$ that is
\begin{eqnarray*}
T^{k}\boldsymbol{\pi}_Q\circ \mathbf{D} \circ \boldsymbol{\gamma} =0.
\end{eqnarray*}
This is due to the calculation (\ref{Disvert}). Now, consider a vector field $Y$ on $Q$. By pushing $Y$ forward by $\boldsymbol{\gamma}$, we arrive at a vector field $\boldsymbol{\gamma}_*Y$ on $T^*_{k,\theta}Q$. We compute the following
\begin{equation}
\begin{split}
\flat(\mathbf{D})(\boldsymbol{\gamma}_*Y)&=\sum_{\kappa=1}^k  \iota_{D_a}\Omega_\theta^\kappa(\boldsymbol{\gamma}_*Y)=\sum_{\kappa=1}^k  \Omega_\theta^\kappa(X_{\kappa}-\boldsymbol{\gamma}_*X_{\kappa}^\gamma,\boldsymbol{\gamma}_*Y)\\
&=\sum_{\kappa=1}^k  \Omega_\theta^\kappa(X_{\kappa},\boldsymbol{\gamma}_*Y)-\sum_{\kappa=1}^k  \Omega_\theta^\kappa(\gamma_*X_{\kappa}^\gamma,\boldsymbol{\gamma}_*Y) \\
&=\sum_{\kappa=1}^k \iota_{X_{\kappa}}\Omega_\theta^\kappa(\boldsymbol{\gamma}_*Y)-\sum_{\kappa=1}^k\boldsymbol{\gamma}^*(\Omega_\theta^\kappa)(X_{\kappa}^\gamma,Y)\\ &=\boldsymbol{\gamma}^*\left(\sum_{\kappa=1}^k\iota_{X_{\kappa}}\Omega_\theta^\kappa\right)(Y) -\sum_{\kappa=1}^k\boldsymbol{\gamma}^*(\Omega_\theta^\kappa)(X_{\kappa}^\gamma,Y)\\
&=(\boldsymbol{\gamma}^*\d_\theta H)(Y) -\sum_{\kappa=1}^k\boldsymbol{\gamma}^*(\Omega_\theta^\kappa)(X_{\kappa}^\gamma,Y)=0
\end{split}
\end{equation}
where we have employed the second condition for the first term at the last line, whereas we employed the closure of $\gamma$  for the second term at the last line. Observing the decomposition of the $k$-vector spaces given by (\ref{decomp-k}), we arrive at that $\mathbf{D}$ is in the kernel of $\flat$ this gives the condition (2).

\end{proof}

\section{An example} \label{Exp}
Let $Q=\mathbf{R}^2-\{\mathbf{0}\}$ be the two dimensional punctured Euclidean space. Here, $\mathbf{0}$ denotes the origin. On this manifold $Q$, consider the following one-form
\begin{equation} \label{Lee-ex}
\vartheta=2\frac{xdy-ydx}{x^2+y^2}
\end{equation}
which is closed but not exact. Consider the $k$-cotangent bundle $T^*_kQ$ over the base manifold $Q$ equipped with the canonical projection $\boldsymbol{\pi}_Q$ and the local coordinates $(x,y,p^\kappa_x,p^\kappa_y)$ where $\kappa$ runs from $1$ to $k$. We pull the one-form $\vartheta$ in \eqref{Lee-ex}  back to $T^*_kQ$ with the projection $\boldsymbol{\pi}_Q$. This results in a one-form $\theta=(\boldsymbol{\pi}_Q)^*\vartheta$ on the $k$-cotangent bundle $T^*_kQ$. The semi-basic character of $\theta$ manifests that its  local realization is the same as the one-form in \eqref{Lee-ex}. 
We define a $k$ number of one-forms and a $k$ number of closed two-forms on $T^*_kQ$ by pulling back the canonical one forms and the symplectic two-forms on each component of $T^*_kQ$, so that we have 
\begin{equation} \label{kappa-forms}
\theta^\kappa=(\pi^\kappa)^*(\theta_Q)=p^\kappa_xdx+dp^\kappa_ydy, \qquad \Omega^\kappa=(\pi^\kappa)^*(\Omega_Q)=dx\wedge dp^\kappa_x+dy\wedge dp^\kappa_y.
\end{equation}
It is immediate to observe that minus the exterior derivative of $\theta^\kappa$ is precisely $\Omega^\kappa$ for all $\kappa$ running from $1$ to $k$. By means of the differential forms exhibited in (\ref{kappa-forms}), we are defining the following family of two-forms 
\begin{equation} \label{Ex-Omega}
\Omega^\kappa_\theta=\Omega^\kappa+\theta\wedge \theta^\kappa= dx\wedge dp^\kappa_x+dy\wedge dp^\kappa_y - 2\frac{yp^\kappa_y+xp^\kappa_x}{x^2+y^2}dx\wedge dy,
\end{equation}
where $\theta$ is the one-form defined as $(\boldsymbol{\pi}_Q)^*\vartheta$. After a direct calculation it is easy to  observe that the set $[\Omega^\kappa_\theta]$ satisfies the first and the second condition in \eqref{condksymp-glo}. To see that this family determines a l.c.k-s. structure we define the integrable distribution 
\begin{equation} \label{loc-intdis}
V=\left \langle \frac{\partial}{\partial p^\kappa_x},  \frac{\partial}{\partial p^\kappa_y}\right \rangle.
\end{equation}  
Referring to this, observation easily validates the third condition in \eqref{condksymp-glo}.
\bigskip

\noindent 
Now we introduce the following quadratic Hamiltonian function 
\begin{equation} \label{Ex-ham-func}
H=\sum_{\kappa=1}^k\frac{1}{2}\left((p^\kappa_x)^2+(p^\kappa_y)^2\right)
\end{equation}
on $T_k^*Q$ and recall the HDW  equation (\ref{Sum2}) in the k.l.c.s. framework defined in terms of the Lichnerowicz-deRham differential $d_\theta$. We compute the Lichnerowicz-deRham differential of the Hamiltonian function as follows 
\begin{equation}
\d_\theta H=\d H-H \theta= \sum_{\kappa=1}^k {p^\kappa_x}\d{p^\kappa_x}+{p^\kappa_y}\d{p^\kappa_y}  - \sum_{\kappa=1}^k\frac{1}{2}((p^\kappa_x)^2+(p^\kappa_y)^2) \frac{xdy-ydx}{x^2+y^2}.
\end{equation}
For a HDW $k$-vector field ${\bf X}=(X_1,...,X_k)$, the HDW equation (\ref{Sum2}) determines the constitutive vector fields 
\begin{equation}
X_\kappa=A_\kappa\frac{\partial}{\partial {x}} + B_\kappa \frac{\partial}{\partial {y}} + 
(C_\kappa)^\lambda \frac{\partial}{\partial {{p}^\lambda_x}}
+ (D_\kappa)^\lambda \frac{\partial}{\partial {p^\lambda_y}},
\end{equation}
with the following explicit computations
\begin{equation}
\begin{split}
&A_\kappa=p^\kappa_x,\quad B_\kappa=p^\kappa_y,\quad 
\sum_{\kappa=1}^k (C_\kappa)^\kappa=\sum_{\kappa=1}^k\frac{y(p^\kappa_y)^2-y(p^\kappa_x)^2+2xp^\kappa_xp^\kappa_y}{x^2+y^2}, \\
&\sum_{\kappa=1}^k (D_\kappa)^\kappa=\sum_{\kappa=1}^k\frac{x(p^\kappa_y)^2-x(p^\kappa_x)^2-2yp^\kappa_xp^\kappa_y}{x^2+y^2}.  
\end{split} 
\end{equation}
\bigskip

\noindent 
Let us now apply the global HJ theorem \eqref{HJ-l.c.k-s.-1} to the present example. Start with a family of one-forms 
\begin{equation}
\gamma^\kappa=\beta^\kappa(x,y)dx+\rho^\kappa(x,y)dy
\end{equation}
satisfying $d_\vartheta \gamma^\kappa=0$ for each $\kappa=1,...,k$ in order to guarantee that the image space of $\boldsymbol{\gamma}$ is a Lagrangian submanifold of the l.c.k-s. manifold $T^*_{k,\theta} Q$ equipped with the two-forms $[\Omega^\kappa_\theta]$ in (\ref{Ex-Omega}). That is, we have 
\begin{equation}
\frac{\partial \rho^\kappa}{\partial x} - \frac{\partial \beta^\kappa}{\partial y} + \frac{2x\beta^\kappa-2y\rho^\kappa}{x^2+y^2}=0, \qquad  \forall \kappa=1,\dots,k.
 \end{equation}
Referring to the third condition in the HJ theorem \eqref{HJ-l.c.k-s.-1}, we write the HJ equation as 
\begin{equation}
d_\vartheta(H\circ \gamma)=\frac{1}{2}d\Big(\sum_{\kappa=1}^k (\beta^\kappa)^2+(\rho^\kappa)^2\Big)-\Big( \sum_{\kappa=1}^k(\beta^\kappa)^2+(\rho^\kappa)^2\Big)\left( \frac{xdy-ydx}{x^2+y^2}\right)=0.
\end{equation}
Explicitly, we compute the following system of equations
\begin{equation}
\begin{split}
\sum_{\kappa=1}^k \Big(\beta^\kappa\frac{\partial \beta^\kappa}{\partial x} + \rho^\kappa\frac{\partial \rho^\kappa}{\partial x} + \Big( (\beta^\kappa)^2+(\rho^\kappa)^2\Big)\frac{y}{x^2+y^2} \Big) &=0,
\\
\sum_{\kappa=1}^k \Big(\beta^\kappa\frac{\partial \beta^\kappa}{\partial y} + \rho^\kappa\frac{\partial \rho^\kappa}{\partial y} - \Big( (\beta^\kappa)^2+(\rho^\kappa)^2\Big)\frac{x}{x^2+y^2}\Big) &=0.
\end{split}
\end{equation}
\bigskip

\noindent 
We write this system in a local coordinate chart. For it, we choose the polar coordinates $x=r\cos\phi$, $y=r\sin\phi$ on an open chart in $Q$, and the induced coordinates $(r,\phi, p^\kappa_r, p^\kappa_\phi)$ on the open chart $U_\alpha$. In these coordinates, the Lee form $\theta$ turns out to be an exact one-form $d\sigma_\alpha$ for $\sigma_\alpha=2\phi$ whereas the two-form in (\ref{Ex-Omega}) reduces to
\begin{equation}
\Omega_\theta^\kappa\vert_\alpha=dr\wedge dp^\kappa_r+d\phi\wedge dp^\kappa_\phi-2p^\kappa_rdr\wedge d\phi.
\end{equation}
Notice that this two-form is not closed but the following one, which is defined according to the formula in (\ref{ua}), 
 \begin{equation}
\Omega^\kappa_\alpha=e^{-\sigma_\alpha}\Omega^\kappa_\theta\vert_\alpha=e^{-2\phi} \Omega^\kappa_\theta\vert_\alpha=e^{-2\phi}\Big(dr\wedge dp^\kappa_r+d\phi\wedge dp^\kappa_\phi-2p^\kappa_rdr\wedge d\phi\Big)
\end{equation}
is closed. One can check that the set $(U_\alpha,[\Omega^\kappa_\alpha],V_\alpha)$ determines a $k$-symplectic manifold. Here, $V_\alpha$ is the restriction of the integrable distribution $V$ presented in \eqref{loc-intdis}. In the local chart, $V_\alpha$ is spanned by tangent vectors
\begin{equation} \label{loc-intdis2}
V=\left \langle \left(\cos\phi\frac{\partial}{\partial p^\kappa_r} -r\sin\phi \frac{\partial}{\partial p^\kappa_\phi} \right) ,  \left(\sin\phi\frac{\partial}{\partial p^\kappa_r} + r\cos\phi \frac{\partial}{\partial p^\kappa_\phi}\right) \right \rangle.
\end{equation} 
The Hamiltonian function (\ref{Ex-ham-func}) can be written as \begin{equation}
H\vert_\alpha=\frac{1}{2}\sum_{\kappa=1}^k \left( (p^\kappa_{r})^2+\frac{1}{r^2}(p^\kappa_\phi)^2 \right)
\end{equation}
and then, according to (\ref{glueHamFunc}), define the local function 
\begin{equation}
H_\alpha=e^{-\sigma_\alpha}H\vert_\alpha=\frac{1}{2}e^{-2\phi}\sum_{\kappa=1}^k \left( (p^\kappa_{r})^2+\frac{1}{r^2}(p^\kappa_\phi)^2  \right)
\end{equation}
on the local chart $U_\alpha$.  We find the Hamiltonian dynamics generated by the local HDW equation 
\begin{equation}
\sum_{\kappa=1}^k\iota_{{(X_\kappa)}_\alpha}\Omega^\kappa_\alpha=dH_\alpha.
\end{equation}
For a vector field
$$X_\kappa=A_\kappa\partial_r+B_\kappa\partial_\phi+(C_\kappa)^\lambda\partial_{p^\lambda_r}+(D_\kappa)^\lambda\partial_{p^\lambda_\phi}$$
we obtain the coefficients
\begin{equation}
A_\kappa=p^\kappa_r, \quad B_\kappa=\frac{1}{r^2}p^\kappa_\phi, \quad \sum^k_{\kappa=1}(C_\kappa)^\kappa=\sum^k_{\kappa=1}\frac{1}{r^2}{p^\kappa_\phi}(2 p^\kappa_r+ \frac{1}{r}p^\kappa_\phi), \quad \sum^k_{\kappa=1}(D_\kappa)^\kappa=\sum^k_{\kappa=1}-({p^\kappa_r})^2+\frac{1}{r^2}(p^\kappa_\phi)^2.
\end{equation}
Now, we are in the realm of the Hamilton-Jacobi theorem (\ref{thm-k-symp-HJ}) valid for $k$-symplectic framework. So, we consider a closed section $\boldsymbol{\gamma}_\alpha=(\gamma_\alpha^1,...\gamma_\alpha^k)$ 
which in local coordinates becomes
$$\gamma^\kappa_\alpha=\xi^\kappa(r,\phi)dr+\eta^\kappa(r,\phi)d\phi.$$
The closure of $\gamma$ implies that for each $\kappa$
\begin{equation}\label{lagsub2}
0=(\gamma^\kappa_\alpha)^*\Omega^\kappa_\alpha=  
e^{-2\phi}\Big(\frac{\partial\xi^\kappa}{\partial\phi}- \frac{\partial\eta^\kappa}{\partial r}-2\xi^\kappa \Big)\d r\wedge\d\phi.
\end{equation}
Now, we calculate $\d (H_\alpha\circ \gamma_\alpha)$ and expect that it vanishes identically. Then, the HJ equations would be 
\begin{equation}
\begin{split}
\sum_{\kappa=1}^k \left(\xi^\kappa\frac{\partial\xi^\kappa}{\partial\phi}+ \eta^\kappa\frac{1}{r^2}\frac{\partial\eta^\kappa}{\partial\phi}-\xi^2-\frac{1}{r^2}(\eta^\kappa)^2\right) =0 \\ \sum_{\kappa=1}^k \left(\xi^\kappa\frac{\partial\xi^\kappa}{\partial r}+  \eta^\kappa\frac{1}{r^2}\frac{\partial\eta^\kappa}{\partial r} -\frac{1}{r^3}(\eta^\kappa)^2\right)=0.
\end{split}
\end{equation}

One can sum these two equations, and simplify the expression to

\begin{equation}
\frac{1}{2}d\left((\xi^{\kappa})^2+\frac{1}{r^2}(\eta^{\kappa})^2\right)=(\xi^{\kappa})^2+\frac{1}{r^2}(\eta^{\kappa})^2
\end{equation}
where ``$d$" denotes total derivative with respect to $(r,\phi)$.

\section{Conclusion and future work}

In this work, we have addressed the glueing problem of HDW equations defined on local $k$-symplectic structures. Accordingly, we have introduced the notion of locally conformal $k$-symplectic (l.c.k-s.) manifolds in Subsection \ref{mot-lcks}. 

In Subsection \ref{Wslcs}, we exhibited the Whitney sum of a $k$ number of cotangent bundles equipped with locally conformal symplectic structures and a HDW equation for this framework has been studied in Subsection \ref{dlcks}. Further, we have depicted the HJ formalism for the HDW theory. To this end, we have first presented a coordinate free proof of the HJ theorem on $k$-symplectic spaces in Subsection \ref{HJ-k}, then we have taken this discussion to l.c.k-s. manifolds in Subsection \ref{HJ-lcks} and we have provided an example. We plan on pursuing investigations on locally conformal $k$-symplectic (l.c.k-s.) manifolds in the following headlines:

\begin{itemize}
\item The reduction of locally conformal symplectic (l.c.s.) manifolds under Lie groups has been studied by many authors see, for example, \cite{HaRy01,St19a, St19b}, whereas reduction on $k$-symplectic manifolds has been given, for example, in \cite{Blaga09,MuReSa04}. In similar fashion, it could be interesting to study the reduction of the HDW equation (\ref{Sum2}) in the l.c.k-s. framework under Lie group symmetry.   

\item The Hamiltonian realization of field theories is also available in the multisymplectic framework \cite{CaIbLeon96,CaIbLeon99,Go91,Ro09}. We find interesting to discuss the glueing problem of local Hamiltonian dynamics in this framework.
\end{itemize}

\section{Acknowledgements}
This work has been partially supported by MINECO grants MTM2016-76-072-P and the ICMAT Severo Ochoa Project SEV-2011-0087 and SEV-2015-0554.

\end{document}